\documentclass[11pt]{article}
\usepackage{microtype}
\usepackage[margin=1in]{geometry}
\usepackage{amsmath}
\usepackage{amsthm}
\usepackage{amssymb}
\usepackage{color}
\usepackage{graphicx}

\newtheorem{proposition}{Proposition}
\newtheorem{lemma}{Lemma}
\newtheorem{theorem}{Theorem}
\newtheorem{fact}{Fact}

\theoremstyle{remark}
\newtheorem{remark}{Remark}

\theoremstyle{definition}
\newtheorem{definition}{Definition}

\newcommand{\homp}[1]{\mathsf{Hom}_{#1}}
\newcommand{\colisop}[1]{\mathsf{ColIso}_{#1}}
\newcommand{\IMM}{\mathsf{IMM}}
\newcommand{\naut}[1]{|\mathit{aut}(#1)|}
\newcommand{\tw}{\mathit{tw}}
\newcommand{\pw}{\mathit{pw}}
\newcommand{\td}{\mathit{td}}
\newcommand{\rat}{\mathbb{Q}}

\newcommand{\classP}{\mathsf{P}}
\newcommand{\NP}{\mathsf{NP}}
\newcommand{\VNP}{\mathsf{VNP}}
\newcommand{\VBP}{\mathsf{VBP}}
\newcommand{\mVF}{\mathsf{mVF}}
\newcommand{\mVP}{\mathsf{mVP}}
\newcommand{\mVNP}{\mathsf{mVNP}}
\newcommand{\mVBP}{\mathsf{mVBP}}
\newcommand{\VP}{\mathsf{VP}}
\newcommand{\SharpP}{\mathsf{\#P}}
\newcommand{\CLIQUE}{\mathsf{CLIQUE}}

\title{Graph Homomorphism Polynomials: Algorithms and Complexity}

\author{
  Balagopal Komarath\thanks{\texttt{bkomarath@rbgo.in}, Saarland University, Saarland Informatics Campus, Germany}
  \and
  Anurag Pandey\thanks{\texttt{apandey@mpi-inf.mpg.de}, Max-Planck Institut f\"ur Informatik, Saarland Informatics Campus, Germany}
  \and
  C. S. Rahul\thanks{\texttt{rahulcs@dubai.bits-pilani.ac.in}, BITS Pilani -- Dubai Campus, UAE}
}

\begin{document}

\maketitle
\thispagestyle{empty}

\begin{abstract}
  We study homomorphism polynomials, which are polynomials that enumerate all homomorphisms from a pattern graph $H$ to $n$-vertex graphs. These polynomials have received a lot of attention recently for their crucial role in several new algorithms for counting and detecting graph patterns, and also for obtaining natural polynomial families which are complete for algebraic complexity classes $\VBP$, $\VP$, and $\VNP$.  We discover that, in the monotone setting, the formula complexity, the ABP complexity, and the circuit complexity of such polynomial families are exactly characterized by the treedepth, the pathwidth, and the treewidth of the pattern graph respectively. 

Furthermore, we establish a single, unified framework, using our characterization, to collect several known results that were obtained independently via different methods. For instance, we attain superpolynomial separations between circuits, ABPs, and formulas in the monotone setting, where  the polynomial families separating the classes all correspond to well-studied combinatorial problems. Moreover, our proofs  rediscover fine-grained separations between these models for constant-degree polynomials. The characterization additionally yields new space-time efficient algorithms for several pattern detection and counting problems.
\end{abstract}

\newpage

\setcounter{page}{1}
\section{Introduction}

This work is a culmination of exploration of four themes in combinatorics, algorithm design, and algebraic complexity -- graph algorithms, homomorphism polynomials, graph parameters, and monotone computations. While each of these themes are of independent interest, a strong interplay among them has become quite apparent in the recent years, and has lead to several new advancements in algorithms and complexity.

The first theme is of \textit{graph algorithms}, where the algorithms that are relevant to this work are those corresponding to pattern detection and counting. Loosely speaking, in such problems, we look for the ``presence" of a graph $H$, called the \textit{pattern graph} in another graph $G$, called the \textit{host graph}\footnote{In this paper, the pattern graph and the host graph are always simple, undirected graphs.}. The notions of presence of one graph in another graph that have been the most prevalent are subgraph isomorphism, induced subraph isomorphism, and homomorphism. In detection and counting algorithms for subgraph isomorphism (resp. induced subgraph isomorphism), we want to detect and count subgraphs (resp. induced subraphs) of an $n$-vertex host graph which is \textit{isomorphic} to the pattern graph. Whereas, while looking for the occurrence of the pattern graph in the host graph, if we relax the mapping to allow multiple vertices in the pattern graph to be mapped to a single vertex in the host graph, while preserving the edge relations, we get a \textit{homomorphism} of the pattern graph in the host graph. When the notion is homomorphism, then we are interested in detecting and counting homomorphisms from the pattern graph to the host graph.

All the above mentioned problems have found many applications, both in theory and practice. For instance, detecting and counting (induced) subgraph isomorphisms are used in extremal graph theory in the study of dense graphs and quasirandom graphs \cite{Chung89, Lovasz12, Lovasz08}, and in many applications which boil down to analyzing real-world graphs. This includes finding protein interactions in biomolecular networks      \cite{AlonDHHS08, Milo824}, finding active chemicals of interest in the research for drug synthesis \cite{Borgelt02}, for advertisement targeting by finding and counting certain social structure patterns in the analysis of social networks \cite{Kong13, zhang}, and also finding user patterns for recommender services on platforms like Amazon and Yelp \cite{ZhaoYLSL17}. The homomorphism counting problem appears for instance, in statistical physics via partition functions, in graph property testing, and extremal graph theory (see \cite{borgs2006counting} for a survey). When both the host graph and the pattern graph are parts of the input, then it can be shown that all these pattern detection problems are $\NP$-complete, since they generalize the $\CLIQUE$ problem, whereas the corresponding counting problems are all known to be $\SharpP$-complete. On the other hand, in almost all the real-world applications pointed above, we have a fixed pattern graph which we are trying to detect or count in a given host graph. Thus, in this work, we focus on the setting when the pattern graph is a fixed graph, and the host graph is a part of the input, Here, since the pattern graph is of a fixed size, say its number of vertices is $k$, all the above problems can be solved using a trivial algorithm, based on exhaustive search, that runs in time $O(n^k)$ and space $O(\log(n))$ where $n$ is the number of vertices in the host graph. However, in almost all the real-world applications pointed above, the host graph is massive, and hence one desires faster algorithms, preferably linear or even sub-linear algorithms. Hence, there is a lot of interest and advances in improving upon this trivial algorithm using ideas from combinatorics, algebraic circuits, and machine learning (see \cite{MarxP14} for a survey, also see \cite{BKS18, Xin2020} and references therein). In fact, these problems are also very interesting from the perspective of complexity theory. For instance, it is conjectured that the best known algorithms for variants of these pattern detection and counting problems can not be further improved. If true, this would imply $\classP \neq \NP$ in a rather strong way (see \cite{LiRR17, Marx10, Rossman18, KR20}). Recently, it was discovered by Curticapean, Dell and Marx \cite{CurticapeanDM17} (also see \cite{DiazST02}) that counting subgraph isomorphisms corresponds to counting linear combinations of homomorphisms, and hence establishing that it is sufficient to just consider the homomorphism counting problem. Several advances have been attained by considering homomorphism problems \cite{BKS18,KloksKM00,KowalukLL13,FloderusKLL15b,WilliamsWWY15}.

Our second theme -- \textit{homomorphism polynomials}, corresponds to one of the most successful ways by which progress has been made towards homomorphism problems, and hence towards (induced) subgraph isomorphism problems too. For a pattern graph $H$ and a host graph $G$, the \textit{homomorphism polynomial} is the polynomial whose monomials enumerate homomorphisms from $H$ to $G$. For using homomorphism polynomials for obtaining new algorithms for the above graph problems for a pattern graph $H$, it suffices to consider the homomorphism polynomial from $H$ to $K_n$, the clique on $n$ vertices. Thus, we call the homomorphism polynomial from $H$ to $K_n$ as the \textit{homomorphism polynomial of $H$} (see Definition \ref{def: hompoly} for a formal definition). It turns out that efficient \textit{arithmetic circuit} (see Section \ref{sec:prelim} for definition) constructions of homomorphism polynomials can be used to obtain almost all known better algorithms for detecting induced subgraph isomorphisms as well (See \cite{BKS18}, also \cite{KloksKM00,KowalukLL13,FloderusKLL15b,WilliamsWWY15}). Therefore, the study of homomorphism polynomials is a crucial area of study within graph pattern detection and counting.  Homomorphism polynomials also turned out to be extremely useful in algebraic complexity theory in the quest for finding natural complete polynomial families for the complexity class $\VP$ (the algebraic complexity analog of $\classP$. See \cite{saptharishi2015survey}). Homomorphism polynomials yield polynomial families that are complete for important algebraic complexity classes such as $\VBP$, $\VP$, and $\VNP$ through a single framework, simply by considering different pattern graphs \cite{DurandMMRS16,MahajanS18, ChauguleLV19}, making it important not only  from the perspective of algorithm design, but also from the perspective of complexity theory. Homomorphism polynomials are thus the central focus of this work too. 

In our third theme, that is of \textit{graph parameters}, the parameters relevant to this work are treewidth, pathwidth, and treedepth of graphs. Treewidth, loosely speaking, is a measure of how far a graph is from being a tree; similarly the pathwidth measures how far a graph is from being a path; and treedepth measures how far a graph is from being a star (see Section \ref{sec:prelim} for definitions). The connection between the parameters treewidth and pathwidth and counting homomorphisms was first explored in \cite{DiazST02} (also see \cite{BKS18,MahajanS18}), where it is shown that when $H$ has bounded treewidth, then there are small-sized arithmetic circuits for homomorphism polynomials, whereas when $H$ has small pathwidth, then the corresponding homomorphism polynomials have small-sized \textit{algebraic branching programs} (ABPs, defined in Section \ref{sec:prelim}). These improved constructions for circuits and ABPs are one of the main source of advancement in finding improved algorithms for these pattern detection and counting problems \cite{BKS18,KloksKM00,KowalukLL13,FloderusKLL15b,WilliamsWWY15}. The connection between the subgraph isomorphism problem and treedepth has also been explored in the context of Boolean computational models in \cite{LiRR17} and \cite{KR20}, where they discuss treewidth and treedepth based upper bounds and lower bounds for counting subgraph isomorphisms for Boolean circuits and formulas. In the context of parameterized counting complexity of these problems, a stronger connection between treewidth and the complexity of counting homomorphisms was established by Dalmau and Jonsson \cite{dalmau04} who showed a dichotomy theorem stating that when the pattern graph has bounded treewidth, then we can count homomorphisms in polynomial time, otherwise, we can show \#W[1]-hardness. Now we turn our attention to the other major application of homomorphism polynomials. It turns out that the the parameters treewidth and pathwidth played a crucial role in the construction of natural complete polynomials for $\VP, \VBP,$ and $\VNP$ too. More specifically, the complete polynomials for $\VP$ and $\VBP$ were homomorphism polynomials corresponding to specific pattern graphs of bounded treewidth and bounded pathwidth respectively. 
Thus, one sees that in the context of these pattern detection and counting problems, and in the homomorphism polynomials of interest, these graph parameters of the pattern graph ubiquitously pop up over and over again as the crucial complexity parameters. This made us wonder, to what extent do these parameters dictate the complexity of these problems and, in particular, the homomorphism polynomials.

Towards this, the starting point of our work is the observation that all the improved arithmetic circuit constructions of homomorphism polynomials based on treewidth and pathwidth, do not use any negative constants in the circuit. This brings us to our final theme, that is, of \textit{monotone computation}.

In the Boolean setting, monotone computations refer to computations that do not involve negation operation, the NOT gate. Similarly, in the algebraic setting for computing polynomials, monotone arithmetic circuits correspond to circuits that do not use negative constants or subtraction gate, that is, there is no cancellation involved in the circuit. Monotone computations are interesting for three main reasons. First reason is that monotone operations have favorable stability properties with respect to rounding errors \cite{Schnorr76}. Secondly, an improved monotone circuit, since it does not use any cancellations and hence algebraic identities, requires combinatorial insights, and often leads to interesting combinatorial algorithms for the problem at hand. Such constructions often reveal interesting combinatorial insights about the problem. For instance, Grenet's monotone algebraic branching program construction of the permanent polynomial \cite{Grenet12}. Finally, the weakness of monotone models resulting from the lack of cancellations, makes them significantly easier to understand. As a consequence, monotone computations are much better understood, both in the Boolean and the algebraic setting, and hence have been used as a starting point towards understanding the general model. Indeed, we know exponential lower bounds as well as separations between important complexity classes when we restrict ourselves to the monotone setting, in contrast to the embarrassingly poor understanding that we have in the general setting. In particular, Schnorr  \cite{Schnorr76} showed an exponential lower bound for the clique polynomial, which happens to be a special case of homomorphism polynomial, where the pattern graph is a clique. Moreover, we know superpolynomial separations between complexity classes\footnote{$\mVF, \mVBP$ and $\mVP$ refer to the class of polynomial families computable by poly-sized monotone arithmetic formulas, monotone algebraic branching programs, and monotone arithmetic circuits, respectively. $\mVNP$ refers to the monotone analog of the complexity class $\VNP$.} $\mVF$ and $\mVBP$ \cite{Snir80}, $\mVBP$ and $\mVP$ \cite{HrubesY16}, and finally, very recently, between $\mVP$ and $\mVNP$ \cite{Yehudayoff19}. None of these separations are settled in the non-monotone setting. Moreover, the best known circuit lower bound in the non-monotone setting is just $\Omega(n \log n)$ \cite{DBLP:journals/tcs/BaurS83}. Thus, apart from the special interests in algorithms based on improved monotone circuits, from the perspective of lower bounds too, it makes sense to first study monotone models, as a first step towards getting lower bounds for the general model. Many times, the best upper bounds also first come in the monotone setting. For instance, the best known upper bound for the permanent polynomial for ABPs is a monotone construction \cite{Grenet12}. In fact, for our central problem of interest, that is, computing homomorphism polynomial, in a lot of settings, the best known upper bounds are still those via monotone constructions\footnote{
We can get the $O(n^{\omega})$ upper bound corresponding to the pattern graphs with treewidth  $= 2$ \cite{CurticapeanDM17}, and $O(n^{k\omega/3})$ for all $k$-vertex graphs \cite{Poljak1985}, where $\omega$ is the exponent of matrix multiplication For large pattern graphs with treewidth $> 2$, the monotone constructions are the best known.}.   

This motivated us to pursue the concrete question of whether these treewidth and pathwidth based improved monotone arithmetic circuits for homomorphisms polynomials can further be improved, and to what extent would these graph parameters dictate it. The answer that we discovered turns out to be conceptually quite satisfying, and settles the problem completely. 

\subsection{Our contributions}
In this work, we fully solve the question of monotone complexity of homomorphism polynomials by showing that they are completely determined by the aforementioned graph parameters. 

For arithmetic circuits, we discover that the treewidth of the pattern graph exactly determines the monotone complexity of its homomorphism polynomial, by establishing that the treewidth based upper bound in \cite{DiazST02, MahajanS18, BKS18} is also a lower bound.
\begin{theorem} \label{thm:mainckt}
  The monotone arithmetic circuit complexity of homomorphism polynomial for a pattern graph $H$ is $\Theta(n^{\tw(H)+1})$, where $\tw(H)$ is the treewidth of $H$.
\end{theorem}

In the case of algebraic branching programs, we find that the pathwidth of the pattern graph is the parameter controlling the monotone complexity of its homomorphism polynomial, again by proving a lower bound that exactly matches the pathwidth based upper bound in \cite{DiazST02, MahajanS18, BKS18}.
\begin{theorem}\label{thm:mainabp}
  The monotone ABP complexity of homomorphism polynomial for a pattern graph $H$ is $\Theta(n^{\pw(H)+1})$, where $\pw(H)$ is the pathwidth of $H$.
\end{theorem}

Finally, for monotone formulas, it is the treedepth of the pattern graph which governs the complexity of its homomorphism polynomial.
\begin{theorem}\label{thm:mainformulas}
  The monotone formula complexity of homomorphism polynomial for a pattern graph $H$ is $\Theta(n^{\td(H)})$, where $\td(H)$ is the treedepth of $H$.
\end{theorem}

Hence, we resolve the question of monotone complexity of homomorphism polynomials completely by showing that treewidth, pathwidth and treedepth exactly characterize the complexity of homomorphism polynomials for arithmetic circuits, ABPs, and arithmetic formulas respectively. This also answers the conceptual question raised earlier asking to what extent do these graph parameters dictate the complexity of these homomorphism polynomials. 

\begin{remark}
It is worth noting that in the Boolean setting, while there are known upper bounds and lower bounds based on these graph parameters for circuits and formulas for (colored) subgraph isomorphism problem, there is still an asymptotic gap between the upper bound and the corresponding lower bounds in these Boolean models (see \cite{LiRR17, KR20}). Moreover, our lower bound techniques are very different from those in the Boolean models.
\end{remark}

The characterization, in addition to giving several new lower bounds, in particular, also allows us to collect, through a unified framework, several classical and recent results for the monotone complexity of polynomials, obtained independently using different methods, over several decades. This is simply because the above theorems imply that for every family of pattern graph with high treewidth, pathwidth, or treedepth, the corresponding homomorphism polynomial family will have high circuit complexity, ABP complexity, or formula complexity, respectively. From the algorithmic perspective, the formula upper bounds on homomorphism polynomials allows us to discover efficient space-time algorithms for several variants of pattern detection and counting.
 
\begin{itemize}
    \item As a first example, applying Theorem \ref{thm:mainckt} and Theorem \ref{thm:mainabp} in very special cases, we reproduce the superpolynomial separation between arithmetic circuits and algebraic branching programs, first discovered by Hrubes and Yehudayoff \cite{HrubesY16}.
    \item  In the same spirit, we also manage to reattain the superpolynomial separation between algebraic branching programs and arithmetic formulas in the monotone setting, using simple applications of Theorem \ref{thm:mainabp} and Theorem \ref{thm:mainformulas}. This was previously known as a consequence of a result by Snir \cite{Snir80}.
    \item Our characterization, in particular, also rediscovers the fine-grained separations between all these models for constant-degree polynomials. That is, for every constant $d$, we get a polynomial family such that it is computable by linear-sized monotone arithmetic circuits, but any monotone ABP computing it must be of size at least $n^d$. Analogously, for every constant $d$, we get a polynomial family computable by linear-sized monotone ABP but any monotone arithmetic formula computing it must be of size at least $n^d$. Earlier, such fine-grained separations could be obtained by applying the results of \cite{Snir80} and \cite{HrubesY16} together.
    \item Another simple application of Theorem \ref{thm:mainckt} yields an exponential lower bound against monotone arithmetic circuits for the clique polynomial $\textrm{CL}_{2n,n}$, which enumerates all cliques of size $n$ in a $2n$-vertex clique, first proved by Schnorr \cite[Theorem 4.4]{Schnorr76}.
    \item  The formula upper bound for homomorphism polynomial yields algorithms for counting homomorphisms that run in time $O(n^{k-1})$ and space $O(\log^2(n))$ for all patterns except cliques. This also implies better algorithms for counting subgraph isomorphisms and detecting induced subgraph isomorphisms using some observations made in \cite{BKS18}. Prior to this, an algorithm that runs in $O(n^{0.174k + o(k)})$ time for counting subgraph isomorphism was known \cite{CurticapeanDM17}. This, however, uses polynomial space in its execution. For a related problem, that is, detecting colored subgraph isomorphisms, such efficient space-time algorithms followed as a consequence of small-sized treedepth based \textit{Boolean} formulas (see \cite{KR20}).
    \end{itemize}

\subsection{Proof ideas}

For our lower bound techniques used in proving Theorem \ref{thm:mainckt}, \ref{thm:mainabp} and \ref{thm:mainformulas}, it turns out that it is easier to work with the so-called colored subgraph isomorphism polynomials instead of the homomorphism polynomials. So, we first show that for proving lower bounds for homomorphism polynomials, it is sufficient to consider the colored subgraph isomorphism polynomial. We establish this by proving that for fixed-size pattern graphs (say, the number of vertices is $k$), their complexity is the same in all the three models that we consider, arithmetic circuits, ABPs, and formulas (see Lemma \ref{lem:homcolor}). 

Now, for the lower bound in Theorem \ref{thm:mainckt}, we successfully answer the question -- \textit{in how many monomial computations can a single gate participate?} Our main technical contribution is to establish that the answer to this question is dictated by the treewidth of the pattern graph (proof of Theorem \ref{thm:cktproof}). For this we start with an arbitrary monomial computation (more precisely, parse trees, discussed in Section \ref{sec:prelim}) in an arbitrary circuit computing the polynomial and use it carefully to construct the so-called \textit{tree decomposition} (Definition \ref{def:treepathdec}) of the pattern graph $H$. Then, if the number of vertices in $H$ is $k$, and its treewidth is $\tw(H)$,  we argue that a gate can participate in the computation of at most $n^{k-\tw(H)-1}$ monomials. We show this by using a weakness of monotone computation that, due to the absence of cancellations, any circuit computing a polynomial cannot compute an invalid submonomial\footnote{an invalid submonomial is a monomial $m$ which does not divide any of the monomial of the target polynomial} at any intermediate gate along the way. This weakness is exploited in almost all known monotone lower bounds, and dates back to Jerrum and Snir \cite{jerrum_snir}. 

For proving the lower bounds in  Theorem \ref{thm:mainabp} and \ref{thm:mainformulas} for formulas and ABPs, we are able to use the same framework as above. Instead of constructing a tree decomposition using the parse trees, we construct a \textit{path decomposition} (Definition \ref{def:treepathdec}) in case of ABPs, and an \textit{elimination tree} (Definition \ref{def:elimination}) in case of formulas. Using the similar weakness of monotone computation, we conclude that the number of  monomials whose computation a gate can participate in is upper bounded by $n^{k-\pw(H)-1}$ in case of ABPs, and by $n^{k-\td(H)}$ in case of formulas, where $\pw(H)$ and $\td(H)$ denote the pathwidth and the treedepth of $H$ respectively.

To obtain the upper bounds claimed in Theorem \ref{thm:mainckt}, \ref{thm:mainabp} and \ref{thm:mainformulas}, we note that the upper bounds for circuits and ABPs that were already known are both monotone constructions, and they go via the tree decomposition and the path decomposition of the pattern graph in case of circuits and ABPs respectively \cite{DiazST02, BKS18, MahajanS18}. We give the formula upper bound using the elimination tree of the pattern graph. A treedepth based monotone formula upper bound is folklore in the Boolean setting. We believe that our formula construction in the arithmetic formula uses similar ideas.

To prove the separation between monotone complexity classes, we observe that pattern graphs with high pathwidth but low treewidth yield superpolynomial separation between circuits and ABPs (discussed in Theorem \ref{thm:cktvsabp}), whereas pattern graphs with high treedepth but low pathwidth give superpolynomial separation between ABPs and formulas (described in Theorem \ref{thm:abpvsformula}). For the first separation, we use a tree as the pattern graph, whereas for the second separation, we use path as the pattern graph.
For an exponential lower bound on the clique polynomial, we simply apply Theorem \ref{thm:mainckt} in combination with the fact that the treewidth of a clique on $n$ vertices is $n-1$.

To get improved efficient space-time algorithms, we notice that our formula construction in the proof of Theorem \ref{thm:mainformulas} is a constant-depth formula and hence can be evaluated in polylog-space and $n^{k-1}$ time (see Theorems \ref{thm:algocount}, \ref{thm:algodetect} and \ref{thm:algotransfer}).

\paragraph{Comparison with previous techniques:}


As mentioned above, a lot of monotone lower bounds (eg.\ \cite{jerrum_snir, HrubesY16}) at their core, use that due to the lack of cancellation, any circuit (resp. ABP or formula) computing a polynomial cannot compute an invalid submonomial at any intermediate gate along the way. 
The central idea behind all the lower bounds that use this weakness is to argue, depending on the polynomials they are working with, that certain gates cannot participate in the computation of too many monomials. It is here that we differ from the known lower bound proofs. Since we are working with homomorphism polynomials, we needed something about homomorphism polynomials that help us argue about it. Here, our main conceptual contribution to the proof technique is a method to construct a tree decomposition (resp. path decomposition and elimination forest) for the pattern graph from a parse tree associated with a monomial in the corresponding homomorphism polynomial. From the perspective of proving superpolynomial lower bounds, comparing it with \cite{HrubesY16, Snir80, Schnorr76}, the key difference between the previous techniques and ours is that all our hard polynomials come from a single framework, that is, from homomorphism polynomials corresponding to various classes of pattern graphs. We find it conceptually satisfying.

For ABPs, Fournier, Malod, Szusterman, and Tavenas \cite{fournier2019} gave a rank-based lower bound method for the monotone setting, inspired by \cite{nisan91} given in the non-commutative setting. However, in their models, the edge labels are homogeneous linear forms, which makes their method unsuitable for the tight lower bounds that we were looking for.

Finally, an approach that seems different from all these aforementioned approaches, including ours, is the one given by Schnorr \cite{Schnorr76}. which also gave an exponential lower bound on the clique polynomial \cite[Theorem 4.4]{Schnorr76}. We show in Section \ref{sec:ckt} that his approach also falls short of proving the lower bound that matches the upper bound. 
In fact, unlike Schnorr's lower bound techniques, our proofs at a very high level, follow an abstract argument that can be shown to always give the optimal lower bound on the number of addition gates (see Theorem \ref{thm:proofuniv}).

It is worth mentioning that our lower bounds, like many other previous lower bounds, depend only on the monomials present in the polynomials, and not on the corresponding coefficients. On the other hand, the proof of separation between monotone $\VP$ and $\VNP$  by Yehudayoff \cite{Yehudayoff19} was sensitive to coefficients too, which was crucial for the separation.

\paragraph{Organization of the paper:}The rest of this paper is organized as follows: Section~\ref{sec:prelim} states the definitions used in this paper. Section~\ref{sec:ckt} proves the circuit lower bounds and Section~\ref{sec:algo} states the algorithms.

To simplify the presentation, we assume that the pattern graph is connected. The complexity for a disconnected pattern is the maximum of the complexity of its connected components.

\section{Preliminaries}
\label{sec:prelim}

For basic notions in graph theory, we refer the readers to \cite{west,Diestel06}.
We first give some definitions that set up our objects of computation and the models of computation.

\begin{definition}
  A polynomial over $\rat$ is called \emph{monotone} if all its coefficients are non-negative.
\end{definition}

Compact representations of polynomials such as the following are usually used by
algorithms.

\begin{definition}
  An \emph{arithmetic circuit} over the variables $x_1,\dotsc,x_n$ is a rooted DAG where each source node (also called an input gate) is labeled by one of the variables $x_i$ or a constant $a\in \rat$. All other nodes (called gates) are labeled with either $+$ (addition) or $\times$ (multiplication). The circuit computes a polynomial over $\rat[x_1,\dotsc,x_n]$ in the usual fashion. The circuit is called \emph{monotone} if all constants are non-negative. The circuit is a \emph{skew circuit} if for all $\times$ gates, at least one of the inputs is a variable or a constant. The circuit is a \emph{formula} if all gates have out-degree at most one. The \emph{size} of a circuit or skew circuit or formula is the number of edges in the circuit. The \emph{depth} of a circuit is the number of gates in the longest path from the root to an input gate.
\end{definition}

Instead of skew circuits, a model that is equivalent in terms of power is usually studied in algebraic complexity.

\begin{definition}
  An \emph{Algebraic Branching Program (ABP)} is a DAG with a unique source node $s$ and a unique sink node $t$. Each edge is labeled with a variable from $x_1,\dotsc,x_n$ or a constant $a\in \rat$. Each path in the DAG from $s$ to $t$ corresponds to a term obtained by multiplying all the edge labels on that path. The polynomial computed by the ABP is the sum of all terms over all paths from $s$ to $t$. The ABP is called \emph{monotone} if all constants are non-negative. The \emph{size} of the ABP is the number of edges.
\end{definition}

It is well-known that for any (monotone) polynomial, the size of the smallest (monotone) ABP and the size of the smallest (monotone) skew circuit are within constant factors of each other. In this paper, we will use the skew circuit definition in our proofs.

In this paper, we look at families of polynomials $(p_n)_{n\geq 0}$ and the optimal size and depth of the models computing them. In this case, the size and depth are functions of $n$ and we are only interested in the asymptotic growth rate of these functions.

For simplicity, when proving lower bounds, we assume that all non-input gates have exactly two incoming edges (both may be from the same gate). This assumption doesn't increase the size of the circuit and may increase the depth by at most a logarithmic factor.

The families of polynomials that we look at in this paper enumerate graph homomorphisms or colored isomorphisms. We first define these notions.

\begin{definition}
  For graphs $H$ and $G$, a \emph{homomorphism} from $H$ to $G$ is a function $\phi: V(H)\mapsto V(G)$ such that $\{i, j\}\in E(H)$ implies $\{\phi(i), \phi(j)\}\in E(G)$. For an edge $e = \{i, j\}$ in $H$, we use $\phi(e)$ to denote $\{\phi(i), \phi(j)\}$.
\end{definition}

\begin{definition}
  Let $H$ be a $k$-vertex graph where its vertices are labeled $[k]$ and let $G$ be a graph where each vertex has a color in $[k]$. Then, a \emph{colored isomorphism} of $H$ in $G$ is a subgraph of $G$ isomorphic to $H$ such that all vertices in the subgraph have different colors and for each edge $\{i, j\}$ in $H$, there is an edge in the subgraph between vertices colored $i$ and $j$.
\end{definition}

Now, we are ready to define our main object of computation, the homomorphism polynomial.

\begin{definition}\label{def: hompoly}
  For a pattern graph $H$ on $k$ vertices, the $n^{\text{th}}$ \emph{homomorphism polynomial} for $H$ is a polynomial on $\binom{n}{2}$ variables $x_e$ where $e = \{u, v\}$ for $u, v\in [n]$.
  
  \begin{equation*}
    \homp{H, n} = \sum_{\phi}\prod_{e} x_{\phi(e)}
  \end{equation*}

  where $\phi$ ranges over all homomorphisms from $H$ to $K_n$ and $e$ ranges over all edges in $H$.
\end{definition}

Next, we define the colored isomorphism polynomial which would be crucial in our proofs. This polynomial enumerates all colored isomorphisms from a pattern to a host graph where there are $n$ vertices of each color. This polynomial can be used to count colored isomorphisms in $n$-vertex host graphs by setting the variables corresponding to edges not in the host graph to $0$.

\begin{definition}
  \label{def:colisop}
  For a pattern graph $H$ on $k$ vertices, the $n^{\text{th}}$ \emph{colored isomorphism  polynomial} for $H$ is a polynomial on $|E(H)|n^2$ variables $x_e$ where $e = \{(i, u), (j, v)\}$ for $u, v\in [n]$ and $\{i, j\}\in E(H)$.
  
  \begin{equation*}
    \colisop{H, n} = \sum_{u_1,\dotsc,u_k}\prod_{i, j} x_{\{(i, u_i), (j, u_j)\}}
  \end{equation*}

  where $u_1,\dotsc,u_k\in[n]$ and $\{i, j\}\in E(H)$.
\end{definition}

We notice that the labeling of $H$ does not affect the complexity of $\colisop{H}$. Given the polynomial $\colisop{H}$ for some labeling of $H$ and if $\psi$ is a relabeling of $H$, then the polynomial $\colisop{H}$ for the new labeling can be obtained by the substitution $x_{\{(i, u), (j, v)\}}\mapsto x_{\{(\psi(i), u), (\psi(j), v)\}}$.

For our main proofs, we need a way to analyze how a monomial is being computed in a circuit, an ABP, or a formula. For this, we use the notion of parse trees.

\begin{definition}
  Let $g$ be a gate in a circuit $C$. A \emph{parse tree} rooted at $g$ is any rooted tree which can be obtained by the following procedure, duplicating gates in $C$ as necessary to preserve the tree structure.

  \begin{enumerate}
  \item The gate $g$ is the root of the tree.
  \item If there is a multiplication gate $g$ in the tree, include all its children in the circuit as its children in the tree.
  \item If there is an addition gate $g$ in the tree, pick an arbitrary child of $g$ in the circuit and include it in the tree.
  \end{enumerate}
\end{definition}

If the root gate of a parse tree is not mentioned, then it is assumed to be the output gate of the circuit. A parse tree witnesses the computation of some term. We note that if $C$ is a formula, then any gate can occur at most once in any parse tree in $C$.

Given a parse tree $T$ that contains a gate $g$, we use $T_g$ to denote the subtree of $T$ rooted at $g$. The tree obtained by removing $T_g$ from $T$ is called the tree outside $T_g$ in $T$. Note that we can replace $T_g$ in $T$ with any parse tree rooted at $g$ to obtain another parse tree. Similarly, if we have two parse trees $T$ and $T'$ that both contain the same multiplication gate $g$ from the circuit, then we can replace the left or right subtree of $T_g$ with the left or right subtree of $T'_g$ to obtain another parse tree. This is because both the left and right child of $g$ in both parse trees are the same and therefore we can apply the aforementioned replacement.

Now, we define the graph parameters most crucial to our work. They are the parameters that turn out to exactly characterize the complexity of the above polynomial families in the models of computations defined above.

\begin{definition}\label{def:treepathdec}
  A \emph{tree decomposition} of $H$ is a tree where each vertex (called a \emph{bag}) in the tree is a subset of vertices of $H$. This tree must satisfy two properties.

  \begin{enumerate}
  \item For every edge $\{i, j\}$ in $H$, there must be at least one bag in the tree that contains both $i$ and $j$.
  \item For any vertex $i$ in $H$, the subgraph of the tree decomposition induced by all bags containing $i$ must be a subtree. This subtree is called the \emph{subtree induced by $i$}.
  \end{enumerate}

  The size of a tree decomposition is the size of the largest bag minus one. The \emph{treewidth} of $H$ is the size of a smallest tree decomposition of $H$.

  A tree decomposition is called a \emph{path decomposition} if it is a path. The \emph{pathwidth} of $H$ is the size of a smallest path decomposition.
\end{definition}

\begin{definition}\label{def:elimination}
  For a connected graph $H$, an \emph{elimination tree} of $H$ is a rooted, directed tree that can be constructed by arbitrarily picking a vertex $u$ in $H$ and adding edges from the roots of elimination trees of connected components of $H-u$ to the root vertex labeled $u$. In particular, if $H$ is a single vertex, then the elimination tree of $H$ is the same single vertex graph.

  The \emph{depth} of an elimination tree is the number of vertices in the longest path from a leaf to the root. The \emph{treedepth} of $H$ is the depth of the smallest depth elimination tree of $H$.
\end{definition}

We note that for any graph $H$, its elimination tree contains exactly $|V(H)|$ vertices. All edges in the tree are directed towards the root and the vertices of the tree are uniquely labeled with the vertices of $H$. If $T$ is an elimination tree for the connected graph $H$, then all edges in $H$ are between vertices that are in an ancestor-descendant relationship in $T$.

We now state some basic facts about treewidth, pathwidth, and treedepth. We combine these facts with the lower bounds in Section~\ref{sec:ckt} to obtain, for any constant $k$, constant-degree polynomial families having linear size circuits (ABPs resp.) but requiring $\Omega(n^k)$ size ABPs (formulas resp.), thus giving us a fine-grained separation between these models. Note that, for constant-degree polynomials, such polynomial factor separations are the best one could ask for between formulas, ABPs, and circuits, since all constant-degree polynomials are computable by polynomial-sized formulas.  Moreover, we use these facts to obtain superpolynomial separations between these models for high degree polynomials.

\begin{fact}
  For all graphs $H$, we have $\tw(H)\leq \pw(H)\leq \td(H)-1$.
\end{fact}

\begin{fact}
  \label{fact:tw}
  For any $p\geq 2$, there is a tree $X_k$ on $k = 2^{p+1} - 1$ vertices that have pathwidth $p$.
\end{fact}

\begin{fact}
  All paths have pathwidth 1 and the $k$-vertex path has treedepth $\lceil \log_2(k+1) \rceil$.
\end{fact}

\section{Algebraic complexity of homomorphism polynomials}
\label{sec:ckt}

In this section, we prove Theorems \ref{thm:mainckt}, \ref{thm:mainabp} and \ref{thm:mainformulas}. thus achieving our main result, that is, the exact characterization for the monotone complexity of homomorphism polynomials. 

We first briefly show the limitation of Schnorr's approach  \cite{Schnorr76} which also serves as a warm up to better familiarize the reader with the homomorphism polynomials before we go into our main proofs. 

Schnorr \cite{Schnorr76} introduced a technique to prove monotone arithmetic circuit lower bounds for polynomials. A \emph{separating set} for a polynomial $p$ is a set of monomials in $p$ such that for any two monomials $s$ and $t$ in the separating set, there does not exist a monomial $m$ of $p$ such that $m$ divides the product $st$. Schnorr showed that the size of a separating set lower bounds the size of the monotone arithmetic circuit computing the polynomial. Using this, he showed that $\homp{K_k}$ requires monotone arithmetic circuits of size $n^k$. We now show that there are pattern graphs for which we cannot prove optimal lower bounds using Schnorr's approach.

\begin{proposition}
  Any separating set for $\homp{P_3,n}$ has size at most $n$ for sufficiently large $n$.
  \label{prop:p3}
\end{proposition}
\begin{proof}
  First note that a homomorphic image of a $P_3$ will either be another $P_3$ or an edge ($P_2$), and consequently, the monomials of $\homp{P_3,n}$ will correspond to either a $P_3$ or an edge. Thus, any separating set $B$ for $\homp{P_3,n}$ contains monomials that correspond to either $P_3$ or edges. Partition $B$ into $B_1$ (set of $P_3$'s) and $B_2$ (set of edges). If two $P_3$'s intersect, then their union must contain a different $P_3$. If two edges intersect, then their union must contain a $P_3$. The product of two monomials that correspond to these intersecting $P_3$'s or edges would be divisible by the monomial that corresponds to the new $P_3$. This shows that $B_1$ and $B_2$ contain at most $n/3$ and $n/2$ monomials respectively.
\end{proof}

Obviously, any circuit (even non-monotone) that computes $\homp{P_3}$ must have $\binom{n}{2}$ size because all variables are present in the polynomial. But, this lower bound is not very interesting. Below, we show that there are patterns with non-trivial monotone lower bounds where Schnorr's approach fails.

\begin{proposition}
  Any separating set for $\homp{C_4,n}$ has size at most $n^2$ for sufficiently large $n$.
\end{proposition}
\begin{proof}
  Since the homomorphic image of a $C_4$ can either be a $C_4$, a $P_3$, or a $P_2$, the monomials in $\homp{C_4,n}$, and subsequently in any separating set, can only contain monomials that correspond to $C_4$, $P_3$, or $P_2$. Now, by the argument in Proposition~\ref{prop:p3}, it remains only to consider monomials that correspond to $C_4$. We claim that any separating set can contain at most $\binom{n}{2}$ such monomials. Observe that for any $x, y\in [n]$, the set can contain at most one $C_4$ where $x$ and $y$ appear as diagonally opposite vertices. If there are two of them, then there is another $C_4$ in the graph obtained by taking the union of those two $C_4$'s. The product of two such monomials would be divisible by the monomial corresponding to this new $C_4$.
\end{proof}

In the rest of this section, we establish optimal lower bounds on the monotone circuit complexity of all homomorphism polynomials.

\begin{lemma}\label{lem:homcolor}
  For any pattern $H$, both the homomorphism polynomial for $H$ and the colored isomorphism polynomial for $H$ have the same monotone arithmetic circuit complexity, monotone ABP complexity, and monotone formula complexity.
\end{lemma}
\begin{proof}
  First, we show how to obtain a circuit that computes $\colisop{H}$ from one that computes $\homp{H}$. We introduce new variables $w_{e}$ for each $e\in E(H)$. Let $C$ be a circuit that computes $\homp{H}$ over the vertex set $[k]\times [n]$. We substitute $x_{\{(i, u), (j, v)\}}$ with $x_{\{(i, u), (j, v)\}}w_{\{i, j\}}$ for all $i, j, u, v$ if $\{i, j\}\in E(H)$. Otherwise, we set the variable to $0$. Let $C'$ be the resulting circuit. We then compute $\frac{\partial^{|E(H)|}}{\partial w_{e_1}\dotsc\partial w_{e_{|E(H)|}}}C'$ using the sum and product rule for partial derivatives. Then, we set $w_e = 0$ for all $e$. The partial differentiation ensures that all monomials must have at least one edge of each color. Setting $w_e = 0$ ensures that all monomials can only have at most one edge of each color. Therefore, the remaining monomials contain each edge in $H$ exactly once. These are exactly those homomorphisms that correspond to colored isomorphisms. For each colored isomorphism, there are $\naut{H}$ ways to relabel the vertices of $H$ to obtain the same edge set. So we divide by $\naut{H}$ to obtain $\homp{H}$ exactly. Each partial differentiation increases the size of the circuit at most by a factor of $3$. Therefore, if $C$ has size $s$, the final circuit has size at most $c3^{|E(H)|}s$ for some constant $c$.

  For the other direction, given a circuit that computes $\colisop{H}$, we can replace each $x_{\{(i, u), (j, v)\}}$ with $x_{\{u, v\}}$ if $u\neq v$ and $0$ otherwise. The circuit now computes $\homp{H}$ because the monomial corresponding to the homomorphism $\phi$ is generated by the corresponding colored isomorphism $i\mapsto (i, \phi(i))$. Also, every colored isomorphism on vertices $(i, u_i)$ for $1\leq i \leq k$ corresponds to a homomorphism to $K_n$ as long as $u_i\neq u_j$ for $\{i, j\}\in E(H)$.

  Notice that both constructions preserve monotonicity and yield a monotone ABP when the original circuit is a monotone ABP.

  A straightforward application of the sum and product rule to compute the partial derivative does not necessarily yield a formula from a formula. While computing the partial derivative, we have to compute both $f$ and $\partial{f}/\partial{x}$ for every sub-formula $f$. If $f = g + h$ for some formulas $g$ and $h$, then we can simply compute $\partial{f}/\partial{x} = \partial{g}/\partial{x} + \partial{h}/\partial{x}$ without using any sub-formula more than once. When the formula $f = gh$ for some formulas $g$ and $h$, we have to compute $\partial{f}/\partial{x} = g\partial{h}/\partial{x} + h\partial{g}/\partial{x}$. Therefore, we are using $g$ and $h$ twice, once for computing $f$ and once for computing $\partial{f}/\partial{x}$. We can convert the resulting circuit into a formula by duplicating the sub-formulas $g$ and $h$. For general formulas, this leads to a quadratic blowup in size. But, for monotone formulas computing constant-degree polynomials, we show that the size increases only by a constant factor when duplicating sub-formulas in this fashion.

  First, we eliminate all gates in the formula that computes $0$ and replace all gates that compute a constant by an input gate labeled with that constant. We call a multiplication gate $f$ \emph{trivial} if $f = ag$ for some constant $a$ and sub-formula $g$. In this case, we have $\partial{f}/\partial{x} = a\partial{g}/\partial{x}$. Therefore, we do not have to make a copy of the non-constant sub-formula $g$. We make a copy of the input gate $a$ here. But, we do not have to make additional copies of this input gate $a$ later for any trivial gate encountered on the path from $f$ to root. This is because $f$ is not a constant and therefore it must be the other input to the trivial gate that is constant. We claim that if the monotone formula computes a polynomial of degree at most $d$, for any gate $g$, there are at most $d$ non-trivial multiplication gates on the path from $g$ to the root gate. This is because each non-trivial gate increases the degree by at least one and no cancellations can occur in a monotone formula. Therefore, we can compute the partial derivative using at most $d+2$ copies of each gate in the original formula. Since the degree of the polynomials we are differentiating is at most $2|E(H)|$ and we perform $|E(H)|$ differentiations, the final formula has size $O\bigl({(2|E(H)|+2)}^{|E(H)|}s\bigr)$ if the original formula has size $s$.
\end{proof}

From now on, we can use either $\colisop{H}$ or $\homp{H}$ to prove our results. For a pattern graph $H$, we say that an edge $\{i, j\}$ or a vertex $i$ in $H$ is present in a monomial over the variables of the colored isomorphism polynomial if there is a variable $x_{\{(i, u), (j, v)\}}$ for some $u$ and $v$ in that monomial.

In our proofs, we restrict our attention to the multiplication gates and the input gates in parse trees. This is because our proofs consider the computation of individual monomials separately and in this case, the addition gates play no role in the computation as all of them have exactly one child in the parse tree.

\begin{theorem}[Theorem \ref{thm:mainckt} restated] \label{thm:cktproof}
  The monotone circuit complexity of $\homp{H}$ is $\Theta\bigl(n^{\tw(H)+1}\bigr)$.
\end{theorem}
\begin{proof}
  The upper bound is already known (See \cite{DiazST02}, \cite{BKS18}).

  Let $H$ be of treewidth $t$. For proving the lower bound, we consider the $n^{\text{th}}$ colored isomorphism polynomial for $H$ (See Definition~\ref{def:colisop}). Consider any monomial $m$ on the vertices $(i, u_i)$ for $1\leq i \leq k$ and its associated parse tree. We assume without loss of generality that both sides of any multiplication gate computes a non-constant term (If they compute a constant, it has to be $1$ and we can simply ignore this subtree). We build a tree decomposition of H from this tree in a bottom-up fashion. In this tree decomposition, each gate is associated with one or two bags in the decomposition. Exactly one of these bags is designated as the \emph{root bag} of that gate.

  An example of this construction is shown in Figure~\ref{fig:tw} where the $4$-vertex cycle is the pattern.

  \begin{enumerate}
  \item For an input gate $x_{\{(i, u), (j, v)\}}$. Add the bag $\{i, j\}$ as a leaf of the tree decomposition. For example, in Figure~\ref{fig:tw}, the input gate labeled $x_{\{(1,u), (2,v)\}}$ is associated with the bag $\{1, 2\}$.

  \item For a multiplication gate $g$. Let $A$ and $B$ be the contents of the root bags of its left and right subtrees. Add a bag containing $A \cup B$ as the parent of those roots. For example, in Figure~\ref{fig:tw}, the gate $g$ is associated with a bag containing $\{1,2,3\}$.
 
  \item If at any gate, there are vertices $(i, u)$ such that the monomial computed at that gate includes all edges incident on $(i, u)$. Then add a new  bag that excludes exactly all such vertices from the current root bag and add it as a parent of the current root bag. This new bag is now considered the root bag associated with the gate $g$. For example, in Figure~\ref{fig:tw}, we observe that at gate $g$, all edges incident on the vertex $(2, v)$ are included and therefore we add a new bag that removes $2$ and make it the root bag associated with $g$.
  \end{enumerate}
  
  The result is a tree decomposition of $H$. All edges of $H$ are covered because all edges of $H$ must be present in the monomial. Since a vertex is forgotten only after all edges incident on it have been multiplied, the subgraph of the tree decomposition induced by any vertex in $H$ is a subtree.

  We consider some gate $g$ in the parse tree that is associated with a bag that contains at least $t+1$ vertices. In the following proof, we only consider the case where the bag contains exactly $t+1$ vertices. If there are more than $t+1$ vertices, then we get a better lower bound. We assume without loss of generality that these vertices are $1,\dotsc,t+1$. We now claim that the gate $g$ can only be present in parse trees of monomials $m'$ such that $m'$ contains vertices $(i, u_i)$ for $1 \leq i \leq t+1$. Suppose for contradiction there is a monomial $m'$ with a parse tree $T'$ that contains $g$ and has vertex $(i, v_i)$ where $v_i\neq u_i$ for some such $i$. Let $T$ be the parse tree for $m$. There are two cases:

  \begin{enumerate}
  \item The vertex $i$ is forgotten at a bag associated with $g$: This means that both left and right subtrees of $T_g$ compute monomials that contain $(i, u_i)$. Now if $T'$ contains $(i, v_i)$ on the left subtree of $T'_g$, then replace right subtree of $T'_g$ with right subtree of $T_g$. Else, the tree $T'$ contains $(i, v_i)$ on the right subtree of $T'_g$ or outside $T'_g$. In both cases we replace the left subtree of $T'_g$ with the left subtree of $T_g$.

  \item Vertex $i$ is not forgotten at a bag associated with $g$: This means that $(i, u_i)$ appears in at least one of the subtrees of $T_g$ and outside $T_g$ in $T$. In $T'$, if $(i, v_i)$ appears in $T'_g$, then replace the tree $T_g$ with $T'_g$ in $T$. Otherwise, the vertex $(i, v_i)$ must appear outside $T'_g$ in $T'$. In this case, replace $T'_g$ with $T_g$ in $T'$.
  \end{enumerate}

  All these new parse trees yield monomials that contain both $(i, u_i)$ and $(i, v_i)$. A contradiction. Therefore, at most $n^{k-t-1}$ monomials of the polynomial can contain $g$ in its parse tree. Since there are $n^k$ monomials, this gives the required lower bound.
\end{proof}

\begin{figure}[ht]
  \centering
  \includegraphics[width=\textwidth]{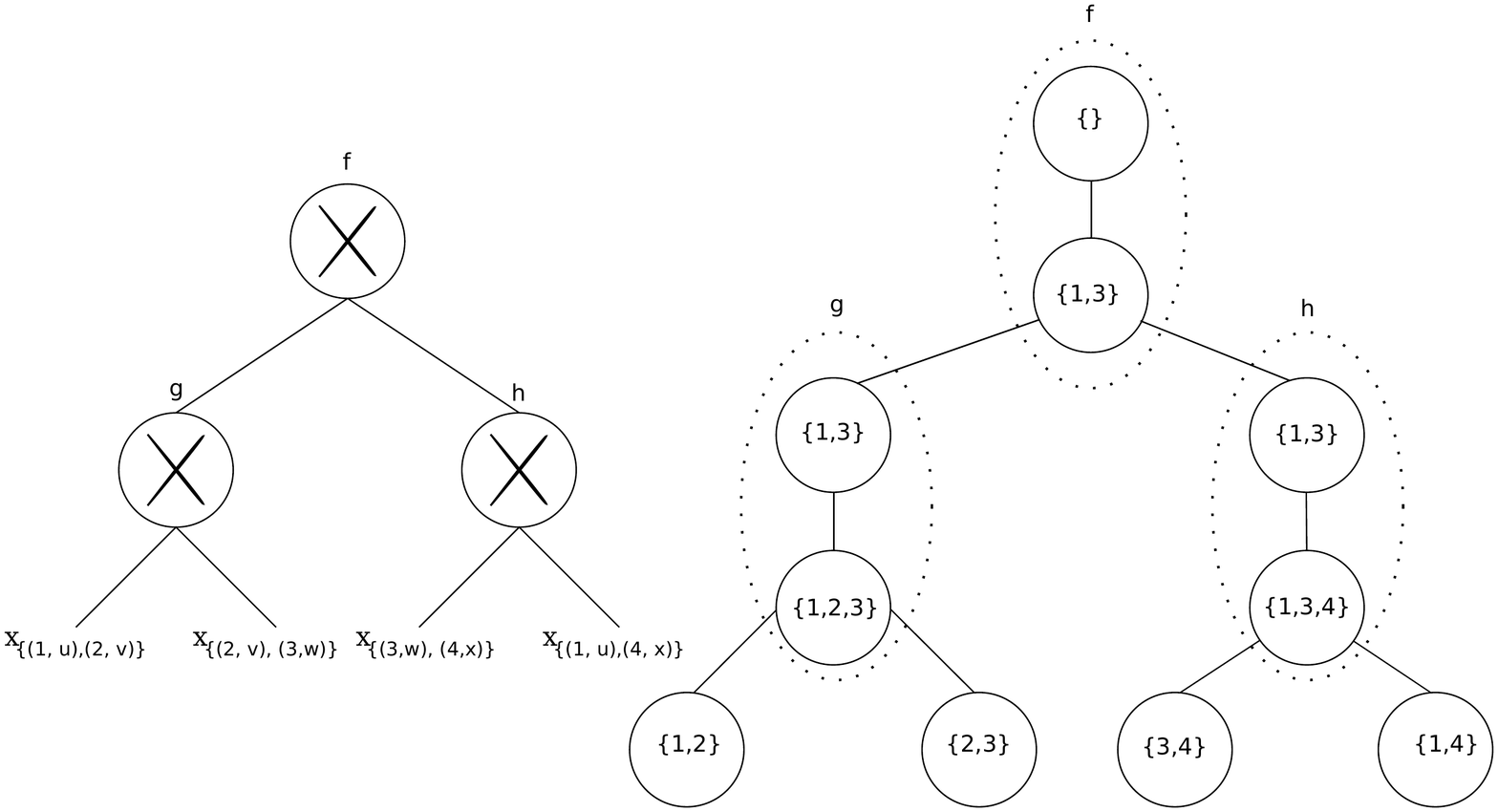}
  \caption{Tree decompositions from parse trees: The pattern is $C_4$ and $u, v, w, x\in [n]$}
  \label{fig:tw}
\end{figure}

\begin{theorem}[Theorem \ref{thm:mainabp} restated] \label{thm:abpproof}
  The monotone ABP complexity of $\homp{H}$ is $\Theta\bigl(n^{\pw(H)+1}\bigr)$.
\end{theorem}
\begin{proof}
  The upper bound is already known (See \cite{DiazST02}, \cite{BKS18}).

  For the lower bound, we modify the proof of Theorem~\ref{thm:cktproof} to obtain a path decomposition instead. We consider the parse tree and start building the path decomposition from a deepest input gate. An example of this construction is shown in Figure~\ref{fig:pw} where the $4$-vertex cycle is the pattern.

  \begin{enumerate}
  \item For the chosen deepest input gate with label $x_{\{(i, u), (j, v)\}}$, add $\{i, j\}$ as a leaf bag. For example, in Figure~\ref{fig:pw}, we choose the input gate labeled $x_{\{(1,u), (2,v)\}}$ as the deepest input gate. A bag containing $\{1,2\}$ is added corresponding to this gate.

  \item For a gate that multiplies $x_{\{(i, u), (j, v)\}}$. Add $i$ and $j$ to the root bag of the other sub-tree and make it the new root. For example, in Figure~\ref{fig:pw}, the gate $g$ multiplies $x_{\{(3,w), (4,x)\}}$ and we add $4$ to the root bag of the gate $h$ to obtain a bag containing $\{1,3,4\}$.

  \item If there are vertices $i$ in $H$ such that all edges incident of $i$ have been multiplied into the monomial, add a new parent bag that forgets all such vertices. This is the new root bag corresponding to the gate. For example, in Figure~\ref{fig:pw}, we observe that at gate $g$, all edges incident on $3$ have been multiplied into the monomial and therefore we add a new parent bag that removes $3$ and make this new bag the root bag of $g$.
  \end{enumerate}

  It is easy to see that the resulting tree is a path. The same proof as for Theorem~\ref{thm:cktproof} shows that this is a path decomposition and that we can find for any parse tree for any monomial $m$, a gate $g$ in the parse tree for $m$ such that at most $n^{k-p-1}$ monomials contain $g$ in its parse tree where $p$ is the pathwidth of the pattern. The lower bound follows.
\end{proof}

\begin{figure}[ht]
  \centering
  \includegraphics[width=\textwidth]{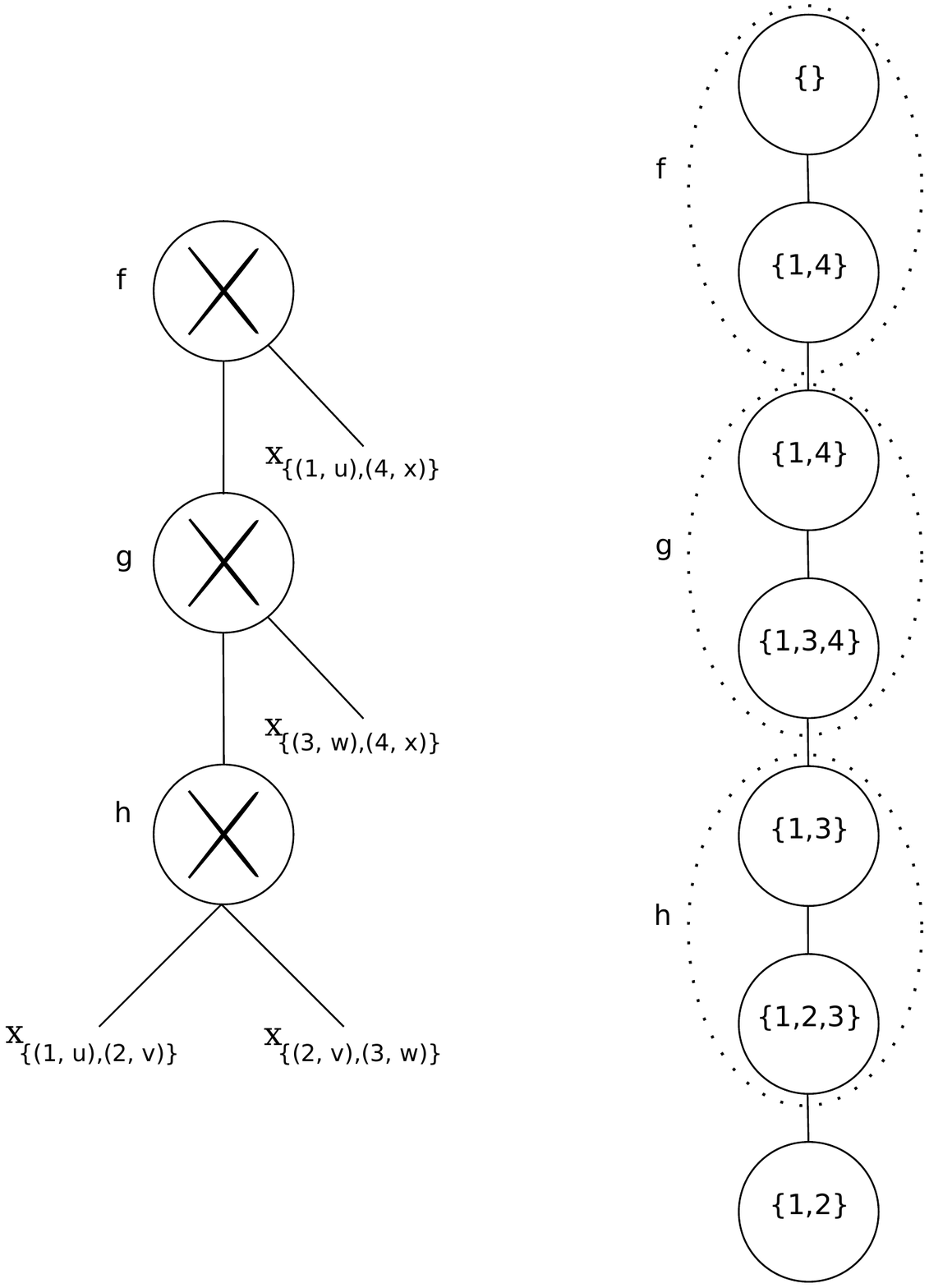}
  \caption{Path decompositions from parse trees: The pattern is $C_4$ and $u, v, w, x\in [n]$}
  \label{fig:pw}
\end{figure}

\begin{theorem}[Theorem \ref{thm:mainformulas} restated] \label{thm:formulaproof}
  The monotone formula complexity of $\homp{H}$ is $\Theta\bigl(n^{\td(H)}\bigr)$.
\end{theorem}
\begin{proof}
  We first prove the upper bound\footnote{We believe this construction is already known as folklore. But we present it here since we use it in our algorithms.}. Let $T$ be an elimination tree of depth $d$ for $H$. We show how to construct a formula of size $n^d$ for $\colisop{H'}$ in a bottom-up fashion where $H'$ is the graph obtained from $T$ by adding all possible edges $\{i, j\}$ where $i$ is an ancestor of $j$ in $T$. The polynomial $\colisop{H}$ can be obtained by setting extra variables in this polynomial to $1$.

  Let $i$ be the label of a node in $T$ such that the path from root to $i$ is labeled $i_1,\dotsc,i_p$ and the children of $i$ are labeled $\ell_1,\dotsc,\ell_s$. We use the notation $(i_j, u_j)_j$ where $j\in [p]$ to denote the $p$ pairs $(i_1, u_1),\dotsc,(i_p, u_p)$. We construct the formula:

  \begin{equation*}
    f^{\{{(i_j, u_j)}_j\}}_i = \sum_u \biggl(\bigl(\prod_j x_{\{(i_j, u_j), (i, u)\}}\bigr) \prod_t f^{\{(i, u), (i_j, u_j)_j\}}_{\ell_t}\biggr)
  \end{equation*}

  where $j\in [p]$ and $t\in [s]$.

  We use induction on the height of the node $i$ in the elimination tree to prove the size upper bound. We claim that the formula that corresponds to such a node has size $O(n^i)$ (constants hidden by the $O$ notation depend only on $H$). For the base case, if $i$ is a leaf node, this formula has size $O(n)$. If $i$ has depth $c$ in $T$, then this formula has size $O(n)\times O(n^{c-1}) = O(n^c)$ by the induction hypothesis.

  If $\{i, j\}$ is not an edge in $H$, then we set all $x_{\{(i, u), (j, v)\}}$ to $1$ in the formula. The formula corresponding to the root node in $T$ is the required polynomial. To prove this, consider the colored subgraph isomorphism containing the vertices $(i, u_i)$ for $1\leq i\leq k$. This monomial has the parse tree obtained by setting the variable $u$ in the outermost summation for each $f^{*}_i$ to $u_i$. For the other direction, the parse tree obtained by setting $u$ to $u_i$ in the outermost summation of each $f^*_i$ generates the monomial that corresponds to the colored subgraph isomorphism on vertices $(i, u_i)$.
  
  We now prove the lower bound. Let $d$ be the treedepth of $H$. We consider a parse tree for a monomial $m$ in a formula computing $\colisop{H}$ and build an elimination tree for $H$ from it. An example of this construction is shown in Figure~\ref{fig:td} where the $4$-vertex path is the pattern. We associate a set of rooted trees with each gate $g$ as follows: If the gate $g$ is the lowest gate in the parse tree such that all edges incident on vertices $i_1,\dotsc,i_r$ are present in the monomial computed at $g$, then make $i_1$ the parent of all roots of trees from the children of $g$. Now, make $i_{j+1}$ the parent of $i_j$ for $1\leq j \leq r-1$. We call the vertices $i_1,\dotsc,i_r$ to be associated with $g$. If there are no such vertices for a gate $g$, then the forest associated with $g$ is simply the union of the forests of its children. We start with empty forests initially. For any edge $\{i, j\}\in E(H)$, the gates that corresponds to $i$ and $j$ in this tree must belong to the path from the input gate for this edge to the root. This shows that this tree is an elimination tree for $H$.

  In Figure~\ref{fig:td}, the input gate labeled $z$ is such that the only edge incident on $(1,u)$ is already multiplied in at $z$. Therefore, we associate the vertex $1$ with $z$ and the set of trees associated with $z$ is just the one-vertex tree $1$. Also, at gate $f$, we have multiplied in all edges incident on $(3,w)$ and therefore, we associate $3$ with $f$. Also, note that this vertex $3$ is the parent of the roots of elimination trees from $g$ and $h$.

 For proving the lower bound, we consider some gate $g$ in the parse tree of $m$ such that $g$ is associated with a leaf at a depth of at least $d$ in the elimination tree. We can assume that the depth is exactly $d$. If it is more, then we obtain a better lower bound. Let this vertex be $d$. Assume without loss of generality that $1,\dotsc, d$ are the vertices on the path from the root to $d$ in the elimination tree. Let $(1, u_1), \dotsc, (d, u_d)$ be the corresponding vertices in $m$. We claim that any monomial $m'$ for which $g$ appears in its parse tree must also have vertices $u_j$ of color $j$ for $1\leq j\leq d$.

  Suppose for contradiction that there is a monomial $m'$ with $g$ in its parse tree and $m'$ has vertex $(i, v_i)$ where $u_i\neq v_i$ for some $1\leq i\leq d$. Let $g'$ be the gate in $T$ such that $i$ is associated with $g'$. Notice that $g'$ must be an ancestor of $g$ ($g'$ could be the same as $g$). Let $T'$ be the parse tree for $m'$. Then, the tree $T'$ must contain $g'$ as well because it contains $g$ and in a formula there is a unique path from any gate to the root. There are two cases:

  \begin{enumerate}
  \item The gate $g'$ is an input gate: In this case, $u_i = v_i$ because a monomial cannot have two different vertices of the same color and $(i, u_i)$ is present in $g'$.

  \item The gate $g'$ is a multiplication gate: In this case, the vertex $(i, u_i)$ must appear in both subtrees of $T_{g'}$. If the vertex $(i, v_i)$ appears in the right (left) subtree of $T'_{g'}$, then we replace the left (right) subtree of $T'_{g'}$ with the left (right) subtree of $T_{g'}$. Otherwise, the vertex $(i, v_i)$ appears outside $T'_{g'}$ in $T'$. In this case, we replace the subtree $T'_{g'}$ with the subtree $T_{g'}$.
  \end{enumerate}

  In all cases, we obtain a monomial that contains both vertices $(i, u_i)$ and $(i, v_i)$. A contradiction. Therefore, at most $n^{k-d}$ monomials of the polynomial can contain the gate $g$ in their parse tree. Since there are a total of $n^k$ monomials, the lower bound follows.

  In Figure~\ref{fig:td}, we can infer using the above argument that the input gate $z$ can be a part of parse trees of at most $n$ different monomials.
\end{proof}

\begin{figure}[ht]
  \centering
  \includegraphics[width=\textwidth]{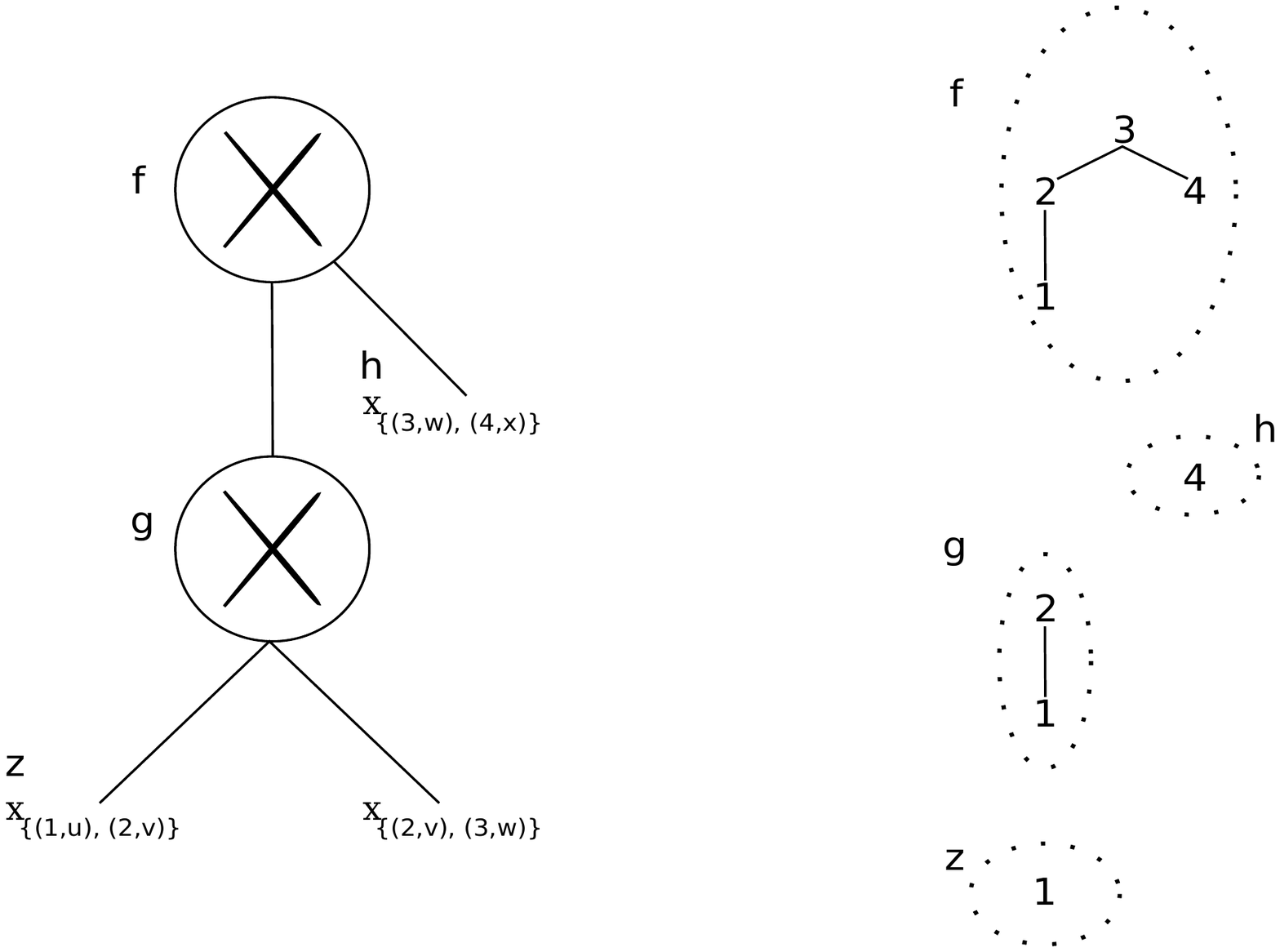}
  \caption{Elimination trees from parse trees: The pattern is $P_4$ and $u, v, w, x\in [n]$}
  \label{fig:td}
\end{figure}

We show how to use our characterizations to prove superpolynomial separations between various monotone models. The polynomials are \emph{not} based on fixed pattern graphs but a natural extension of our polynomials to patterns that grow in size with the host graph.

\begin{theorem}{\normalfont (Compare \cite[Theorem 1, Proposition 13]{HrubesY16})} \label{thm:cktvsabp}
  For any $p\geq 2$ and $n = 2^{p+1} - 1$, there is a polynomial family of degree $n-1$ on $n^3 - n$ variables such that the family has linear size monotone circuits but require $n^{\log_2(n+1)}$ size monotone ABPs.
\end{theorem}
\begin{proof}
  The $n^{th}$ polynomial in the family is simply $\colisop{X_{n}}$. The upper bound follows from the fact that $X_n$ (See Fact~\ref{fact:tw}) is a tree and the construction in \cite{DiazST02} (See also \cite{BKS18}). The lower bound follows from the fact that $\pw(X_n) = \log_2(n+1)$. This is slightly better than the separation shown by Hrubes and Yehudayoff \cite{HrubesY16} because for the polynomial family they use to obtain this separation, they only show super-linear circuit upper bounds and the constant factor in the exponent in the lower bound is worse.
\end{proof}

\begin{theorem}{\normalfont(See \cite[Theorem 3.2]{Snir80})} \label{thm:abpvsformula}
  For $n\geq 2$, there is a polynomial family of degree $n-1$ on $n^3 - n$ variables such that the family has linear size monotone ABPs but require $n^{\lceil \log_2(n+1) \rceil}$ size monotone formulas.
\end{theorem}
\begin{proof}
  The $n^{th}$ polynomial in the family is simply $\colisop{P_{n}}$. The upper bound follows from the fact that $P_n$ has pathwidth $1$ and the construction in \cite{DiazST02} (See also \cite{BKS18}). The lower bound follows from the fact that $\td(P_n) = \lceil \log_2(n+1) \rceil$. We note that Snir \cite{Snir80} used the same polynomial by looking at $\IMM$ although he only considered one entry of $\IMM_{n,n\times n}$. The polynomial $\IMM_{n,n\times n}$ is the $(1,1)^{\text{th}}$ entry of the matrix obtained by multiplying $n$ matrices of order $n\times n$ where each entry of each matrix is a fresh variable. The polynomial $\colisop{P_n}$ is the sum of all entries.
\end{proof}

We now do an analysis of lower bound techniques that work by showing that a ``large'' number of monomials have to be distributed across a ``large'' number of addition gates. We also show that unlike Schnorr's method, we can always lower bound the number of addition gates in a monotone circuit using such an approach. The following proof is an adaptation of the corresponding proof for Boolean circuits \cite{J}.

\begin{theorem}\label{thm:proofuniv}
  Let $p$ be a polynomial such that any monotone circuit $C$ computing it must have $s$ addition gates. Then, we prove that given $C$, we can construct a bijection $b : X \mapsto G$, where $X$ is a set of $s$ monomials in $p$ and $G$ is the set of all addition gates in $C$ and each monomial is mapped to an addition gate in some parse tree for the monomial.
\end{theorem}
\begin{proof}
  Suppose that $C$ is optimal, i.e., has exactly $s$ addition gates. Construct a bipartite graph where one side is the set of all addition gates in $C$ and the other side is the set of all monomials in $p$. Add edge $\{g, m\}$ if and only if $g$ appears in some parse tree for $m$. Now, a bijection $b$ can be constructed from a matching of size $s$ in this graph. By Hall's theorem, the only way this can fail to exist is if there are $k$ addition gates such that only $k-1$ or fewer monomials have parse trees that contain these gates (say $m_1, \dotsc, m_{k-1}$). If $k = 1$, then this means there is an unused addition gate in the circuit. This contradicts the assumption that $C$ is optimal. Otherwise, remove these $k$ gates from $C$. The polynomial computed by this new circuit is $p - \sum_{i=1}^{k-1} c_i m_i$ for some constants $c_i$. But $\sum_{i=1}^{k-1} c_i m_i$ can be computed using $k-2$ addition gates. We will need one more addition gate to add this back to get the original polynomial $p$. This new circuit has only $s - k + k - 2 + 1 = s - 1$ addition gates, a contradiction to the optimality of $C$.
\end{proof}

Therefore, one can always build a lower bound proof by constructing the set of monomials $X$ and an appropriate function $b$ from an arbitrary optimal circuit for the given polynomial. However, unlike Schnorr's method, this approach is not very concrete. The method used to come up with $X$ and $b$ may vary considerably depending on the polynomial.

\section{Algorithms}
\label{sec:algo}

In this section, we obtain algorithms for subgraph isomorphism problems that are simultaneously time and space-efficient. All algorithms follow from theorems in Bl\"aser, Komarath, and Sreenivasaiah \cite{BKS18} and the observation that all constructions in the previous section can be performed in logarithmic space. Evaluating the formulas for homomorphism polynomials require only $O(\log^2(n))$ space because they are logarithmic depth and all values are polynomially bounded. We note that fast matrix multiplication based algorithms yield better running times for all problems considered in this section. However, such algorithms use polynomial space. Algorithms that do not use fast matrix multiplication, called \emph{combinatorial} algorithms, are also widely studied (See \cite{VW09,FominLRSR12,CurticapeanDM17,BKS18}). We use the following well-known fact.

\begin{fact}
  All graphs other than $K_k$ on $k$ or fewer vertices has treedepth at most $k-1$.
\end{fact}

\begin{theorem}\label{thm:algocount}
  For any $k$-vertex pattern graph $H$ other than $K_k$, there is a combinatorial algorithm that runs in time $O\bigl(n^{k-1}\bigr)$ and space $O(\log^2(n))$ that computes the number of subgraphs isomorphic to the pattern graph in a given $n$-vertex host graph.
\end{theorem}
\begin{proof}
  The key idea is that the number of subgraphs isomorphic to $H$ can be expressed as a linear combination of the number of homomorphisms of $H$ and patterns on fewer than $k$ vertices. More details can be found in Theorem~7.3 of \cite{BKS18}.
\end{proof}

\begin{theorem}\label{thm:algodetect}
  For any $k$-vertex pattern graph $H$ other than $K_k$, there is a combinatorial algorithm that runs in time $O\bigl(n^{k-1}\bigr)$ and space $O(\log^2(n))$ that decides whether the pattern graph appears as an induced subgraph in a given $n$-vertex host graph.
\end{theorem}
\begin{proof}
  The key idea is that the number of induced subgraphs isomorphic to $H$ can be expressed as a linear combination of the number of subgraphs isomorphisms from $k$-vertex supergraphs of $H$. We can compute all of these numbers except for $K_k$ in the given time and space. The coefficient of the number of subgraphs isomorphisms of $K_k$ in this linear combination is always a number greater than $1$. Therefore, by choosing a suitable prime $p$, we can count the number of induced subgraphs isomorphic to $H$ modulo $p$ in the given time and space. The induced subgraph isomorphism problem randomly reduces to this problem using a simple reduction that only takes linear time and logarithmic space. More details can be found in Corollary~7.6 of \cite{BKS18}.
\end{proof}

We prove a theorem that connects algorithms that count induced subgraph isomorphisms for any $k$-vertex pattern to algorithms that count $k$-cliques. The motivation behind this theorem is the question posed by \cite{FloderusKLL15} asking whether some induced patterns are easier than others.

\begin{theorem}\label{thm:algotransfer}
  If there is a $k$ vertex pattern graph such that we can count the number of induced subgraph isomorphisms from the pattern to a given $n$ vertex host graph in time $t(n)$ and space $s(n)$, then we can count the number of $k$-cliques in a given $n$ vertex host graph in time $O(t(n) + n^{k-1})$ and space $O(s(n) + \log^2(n))$.
\end{theorem}
\begin{proof}
  The number of induced subgraph isomorphisms for any $k$-vertex pattern $H$ can be expressed as a linear combination of the number of subgraph isomorphisms of $k$-vertex supergraphs of $H$. We can compute every number in the linear combination except the number of subgraph isomorphisms of $K_k$ in $O(n^{k-1})$ time and $O(\log^2(n))$ space. Given that we can also count the number of induced subgraph isomorphisms of $H$ in $t(n)$ time and $s(n)$ space, the result follows. More details can be found in Theorem~7.5 of \cite{BKS18}.
\end{proof}

\bibliographystyle{alpha}
\bibliography{hom}

\newcommand{\etalchar}[1]{$^{#1}$}
\begin{thebibliography}{WWWY15}

\bibitem[ADH{\etalchar{+}}08]{AlonDHHS08}
Noga Alon, Phuong Dao, Iman Hajirasouliha, Fereydoun Hormozdiari, and
  S{\"{u}}leyman~Cenk Sahinalp.
\newblock Biomolecular network motif counting and discovery by color coding.
\newblock In {\em Proceedings 16th International Conference on Intelligent
  Systems for Molecular Biology (ISMB), Toronto, Canada, July 19-23, 2008},
  pages 241--249, 2008.

\bibitem[BB02]{Borgelt02}
C.~{Borgelt} and M.~R. {Berthold}.
\newblock Mining molecular fragments: finding relevant substructures of
  molecules.
\newblock In {\em 2002 IEEE International Conference on Data Mining, 2002.
  Proceedings.}, pages 51--58, 2002.

\bibitem[BCL{\etalchar{+}}06]{borgs2006counting}
Christian Borgs, Jennifer Chayes, L{\'a}szl{\'o} Lov{\'a}sz, Vera~T S{\'o}s,
  and Katalin Vesztergombi.
\newblock Counting graph homomorphisms.
\newblock In {\em Topics in discrete mathematics}, pages 315--371. Springer,
  2006.

\bibitem[BKS18]{BKS18}
Markus Bl{\"{a}}ser, Balagopal Komarath, and Karteek Sreenivasaiah.
\newblock Graph pattern polynomials.
\newblock In Sumit Ganguly and Paritosh~K. Pandya, editors, {\em 38th {IARCS}
  Annual Conference on Foundations of Software Technology and Theoretical
  Computer Science, {FSTTCS} 2018, December 11-13, 2018, Ahmedabad, India},
  volume 122 of {\em LIPIcs}, pages 18:1--18:13. Schloss Dagstuhl -
  Leibniz-Zentrum f{\"{u}}r Informatik, 2018.

\bibitem[BS83]{DBLP:journals/tcs/BaurS83}
Walter Baur and Volker Strassen.
\newblock The complexity of partial derivatives.
\newblock {\em Theor. Comput. Sci.}, 22:317--330, 1983.

\bibitem[CDM17]{CurticapeanDM17}
Radu Curticapean, Holger Dell, and D{\'{a}}niel Marx.
\newblock Homomorphisms are a good basis for counting small subgraphs.
\newblock In Hamed Hatami, Pierre McKenzie, and Valerie King, editors, {\em
  Proceedings of the 49th Annual {ACM} {SIGACT} Symposium on Theory of
  Computing, {STOC} 2017, Montreal, QC, Canada, June 19-23, 2017}, pages
  210--223. {ACM}, 2017.

\bibitem[CGW89]{Chung89}
F.~R.~K. Chung, R.~L. Graham, and R.~M. Wilson.
\newblock Quasi-random graphs.
\newblock {\em Combinatorica}, 9(4):345--362, 1989.

\bibitem[CLV19]{ChauguleLV19}
Prasad Chaugule, Nutan Limaye, and Aditya Varre.
\newblock Variants of homomorphism polynomials complete for algebraic
  complexity classes.
\newblock In {\em Computing and Combinatorics - 25th International Conference,
  {COCOON} 2019, Xi'an, China, July 29-31, 2019, Proceedings}, volume 11653 of
  {\em Lecture Notes in Computer Science}, pages 90--102. Springer, 2019.

\bibitem[Die06]{Diestel06}
R.~Diestel.
\newblock {\em Graph Theory}.
\newblock Electronic library of mathematics. Springer, 2006.

\bibitem[DJ04]{dalmau04}
V\'{\i}ctor Dalmau and Peter Jonsson.
\newblock The complexity of counting homomorphisms seen from the other side.
\newblock {\em Theoret. Comput. Sci.}, 329(1-3):315--323, 2004.

\bibitem[DMM{\etalchar{+}}16]{DurandMMRS16}
Arnaud Durand, Meena Mahajan, Guillaume Malod, Nicolas de~Rugy{-}Altherre, and
  Nitin Saurabh.
\newblock Homomorphism polynomials complete for {VP}.
\newblock {\em Chic. J. Theor. Comput. Sci.}, 2016, 2016.

\bibitem[DST02]{DiazST02}
Josep D{\'{\i}}az, Maria~J. Serna, and Dimitrios~M. Thilikos.
\newblock Counting h-colorings of partial k-trees.
\newblock {\em Theor. Comput. Sci.}, 281(1-2):291--309, 2002.

\bibitem[FKLL15a]{FloderusKLL15b}
Peter Floderus, Miroslaw Kowaluk, Andrzej Lingas, and Eva{-}Marta Lundell.
\newblock Detecting and counting small pattern graphs.
\newblock {\em {SIAM} J. Discrete Math.}, 29(3):1322--1339, 2015.

\bibitem[FKLL15b]{FloderusKLL15}
Peter Floderus, Miroslaw Kowaluk, Andrzej Lingas, and Eva{-}Marta Lundell.
\newblock Induced subgraph isomorphism: Are some patterns substantially easier
  than others?
\newblock {\em Theor. Comput. Sci.}, 605:119--128, 2015.

\bibitem[FLR{\etalchar{+}}12]{FominLRSR12}
Fedor~V. Fomin, Daniel Lokshtanov, Venkatesh Raman, Saket Saurabh, and
  B.~V.~Raghavendra Rao.
\newblock Faster algorithms for finding and counting subgraphs.
\newblock {\em J. Comput. Syst. Sci.}, 78(3):698--706, 2012.

\bibitem[FMST19]{fournier2019}
Herv{\'e} Fournier, Guillaume Malod, Maud Szusterman, and S{\'e}bastien
  Tavenas.
\newblock {Nonnegative Rank Measures and Monotone Algebraic Branching
  Programs}.
\newblock In Arkadev Chattopadhyay and Paul Gastin, editors, {\em 39th IARCS
  Annual Conference on Foundations of Software Technology and Theoretical
  Computer Science (FSTTCS 2019)}, volume 150 of {\em Leibniz International
  Proceedings in Informatics (LIPIcs)}, pages 15:1--15:14, Dagstuhl, Germany,
  2019. Schloss Dagstuhl--Leibniz-Zentrum fuer Informatik.

\bibitem[Gre12]{Grenet12}
Bruno Grenet.
\newblock An upper bound for the permanent versus determinant problem, 2012.

\bibitem[HY16]{HrubesY16}
Pavel Hrubes and Amir Yehudayoff.
\newblock On isoperimetric profiles and computational complexity.
\newblock In {\em 43rd International Colloquium on Automata, Languages, and
  Programming, {ICALP} 2016, July 11-15, 2016, Rome, Italy}, pages 89:1--89:12,
  2016.

\bibitem[JS82]{jerrum_snir}
Mark Jerrum and Marc Snir.
\newblock Some exact complexity results for straight-line computations over
  semirings.
\newblock {\em J. ACM}, 29(3):874–897, July 1982.

\bibitem[KKM00]{KloksKM00}
Ton Kloks, Dieter Kratsch, and Haiko M{\"{u}}ller.
\newblock Finding and counting small induced subgraphs efficiently.
\newblock {\em Inf. Process. Lett.}, 74(3-4):115--121, 2000.

\bibitem[KLL13]{KowalukLL13}
Miroslaw Kowaluk, Andrzej Lingas, and Eva{-}Marta Lundell.
\newblock Counting and detecting small subgraphs via equations.
\newblock {\em {SIAM} J. Discrete Math.}, 27(2):892--909, 2013.

\bibitem[KR20]{KR20}
Deepanshu Kush and Benjamin Rossman.
\newblock Tree-depth and the formula complexity of subgraph isomorphism.
\newblock {\em CoRR}, abs/2004.13302, 2020.

\bibitem[KZY13]{Kong13}
Xiangnan Kong, Jiawei Zhang, and Philip~S. Yu.
\newblock Inferring anchor links across multiple heterogeneous social networks.
\newblock In {\em Proceedings of the 22nd ACM International Conference on
  Information \& Knowledge Management}, CIKM '13, page 179–188, New York, NY,
  USA, 2013. Association for Computing Machinery.

\bibitem[Lov12]{Lovasz12}
L\'{a}szl\'{o} Lov\'{a}sz.
\newblock {\em Large networks and graph limits}, volume~60 of {\em American
  Mathematical Society Colloquium Publications}.
\newblock American Mathematical Society, Providence, RI, 2012.

\bibitem[LPH{\etalchar{+}}20]{Xin2020}
Xin Liu, Haojie Pan, Mutian He, Yangqiu Song, Xin Jiang, and Lifeng Shang.
\newblock Neural subgraph isomorphism counting.
\newblock In {\em Proceedings of the 26th ACM SIGKDD International Conference
  on Knowledge Discovery \& Data Mining}, KDD '20, page 1959–1969, New York,
  NY, USA, 2020. Association for Computing Machinery.

\bibitem[LRR17]{LiRR17}
Yuan Li, Alexander~A. Razborov, and Benjamin Rossman.
\newblock On the ac\({}^{\mbox{0}}\) complexity of subgraph isomorphism.
\newblock {\em {SIAM} J. Comput.}, 46(3):936--971, 2017.

\bibitem[LS08]{Lovasz08}
L\'{a}szl\'{o} Lov\'{a}sz and Vera~T. S\'{o}s.
\newblock Generalized quasirandom graphs.
\newblock {\em J. Combin. Theory Ser. B}, 98(1):146--163, 2008.

\bibitem[Mar10]{Marx10}
D\'{a}niel Marx.
\newblock Can you beat treewidth?
\newblock {\em Theory Comput.}, 6:85--112, 2010.

\bibitem[MP14]{MarxP14}
D{\'{a}}niel Marx and Michal Pilipczuk.
\newblock Everything you always wanted to know about the parameterized
  complexity of subgraph isomorphism (but were afraid to ask).
\newblock In Ernst~W. Mayr and Natacha Portier, editors, {\em 31st
  International Symposium on Theoretical Aspects of Computer Science {(STACS}
  2014), {STACS} 2014, March 5-8, 2014, Lyon, France}, volume~25 of {\em
  LIPIcs}, pages 542--553. Schloss Dagstuhl - Leibniz-Zentrum f{\"{u}}r
  Informatik, 2014.

\bibitem[MS18]{MahajanS18}
Meena Mahajan and Nitin Saurabh.
\newblock Some complete and intermediate polynomials in algebraic complexity
  theory.
\newblock {\em Theory Comput. Syst.}, 62(3):622--652, 2018.

\bibitem[MSOI{\etalchar{+}}02]{Milo824}
R.~Milo, S.~Shen-Orr, S.~Itzkovitz, N.~Kashtan, D.~Chklovskii, and U.~Alon.
\newblock Network motifs: Simple building blocks of complex networks.
\newblock {\em Science}, 298(5594):824--827, 2002.

\bibitem[Nis91]{nisan91}
Noam Nisan.
\newblock Lower bounds for non-commutative computation.
\newblock In {\em Proceedings of the Twenty-Third Annual ACM Symposium on
  Theory of Computing}, STOC '91, page 410–418, New York, NY, USA, 1991.
  Association for Computing Machinery.

\bibitem[NP85]{Poljak1985}
Jaroslav Nešetřil and Svatopluk Poljak.
\newblock On the complexity of the subgraph problem.
\newblock {\em Commentationes Mathematicae Universitatis Carolinae},
  026(2):415--419, 1985.

\bibitem[Ros18]{Rossman18}
Benjamin Rossman.
\newblock Lower bounds for subgraph isomorphism.
\newblock In {\em Proceedings of the {I}nternational {C}ongress of
  {M}athematicians---{R}io de {J}aneiro 2018. {V}ol. {IV}. {I}nvited lectures},
  pages 3425--3446. World Sci. Publ., Hackensack, NJ, 2018.

\bibitem[Sap15]{saptharishi2015survey}
Ramprasad Saptharishi.
\newblock A survey of lower bounds in arithmetic circuit complexity.
\newblock {\em Github survey}, 2015.

\bibitem[Sar]{J}
Jayalal Sarma.
\newblock Personal communication.

\bibitem[Sch76]{Schnorr76}
C.P. Schnorr.
\newblock A lower bound on the number of additions in monotone computations.
\newblock {\em Theoretical Computer Science}, 2(3):305 -- 315, 1976.

\bibitem[Sni80]{Snir80}
Marc Snir.
\newblock On the size complexity of monotone formulas.
\newblock In Jaco de~Bakker and Jan van Leeuwen, editors, {\em Automata,
  Languages and Programming}, pages 621--631, Berlin, Heidelberg, 1980.
  Springer Berlin Heidelberg.

\bibitem[VW09]{VW09}
Virginia Vassilevska and Ryan Williams.
\newblock Finding, minimizing, and counting weighted subgraphs.
\newblock In {\em Proceedings of the Forty-First Annual ACM Symposium on Theory
  of Computing}, STOC '09, page 455–464, New York, NY, USA, 2009. Association
  for Computing Machinery.

\bibitem[Wes00]{west}
Douglas~B. West.
\newblock {\em Introduction to Graph Theory}.
\newblock Prentice Hall, 2 edition, September 2000.

\bibitem[WWWY15]{WilliamsWWY15}
Virginia~Vassilevska Williams, Joshua~R. Wang, Richard~Ryan Williams, and
  Huacheng Yu.
\newblock Finding four-node subgraphs in triangle time.
\newblock In Piotr Indyk, editor, {\em Proceedings of the Twenty-Sixth Annual
  {ACM-SIAM} Symposium on Discrete Algorithms, {SODA} 2015, San Diego, CA, USA,
  January 4-6, 2015}, pages 1671--1680. {SIAM}, 2015.

\bibitem[Yeh19]{Yehudayoff19}
Amir Yehudayoff.
\newblock Separating monotone vp and vnp.
\newblock In {\em Proceedings of the 51st Annual ACM SIGACT Symposium on Theory
  of Computing}, STOC 2019, page 425–429, New York, NY, USA, 2019.
  Association for Computing Machinery.

\bibitem[ZW86]{zhang}
J.~Zhang and G.~Wu.
\newblock Targeting social advertising to friends of users who have interacted
  with an object associated with the advertising, Dec. 15 2010. US Patent App.
  12/968,786.

\bibitem[ZYL{\etalchar{+}}17]{ZhaoYLSL17}
Huan Zhao, Quanming Yao, Jianda Li, Yangqiu Song, and Dik~Lun Lee.
\newblock Meta-graph based recommendation fusion over heterogeneous information
  networks.
\newblock In {\em KDD}, pages 635--644, 2017.

\end{thebibliography}

\end{document}